%% file: sequential_resubmit.tex
\documentclass[10pt,aps,pra,twocolumn,%showpacs,
amsmath,amssymb,superscriptaddress]{revtex4-1}
\usepackage{graphicx,verbatim}
\usepackage{amssymb,amsthm,amsfonts,amstext,color}
\usepackage[english]{babel}
\usepackage[utf8]{inputenc}

\usepackage{hyperref}

\newcommand{\ket}[1]{\left| #1 \right\rangle}
\newcommand{\bra}[1]{\left\langle #1 \right|}

\newcommand{\proj}[1]{\ket{#1}\!\bra{#1}}

\newcommand{\ev}[1]{\left\langle #1\right\rangle}

\DeclareMathOperator{\trace}{tr}
\DeclareMathOperator{\identity}{\mathbb{I}}

\newtheorem{thm}{Theorem}
\newtheorem{prob}{Open Problem}

\newtheorem{prop}[thm]{Proposition}

\theoremstyle{definition}
\newtheorem{dfn}{Definition}

\newcommand{\rodrigo}[1]{{\color{black} #1}}

\begin{document}

\title{Nonlocality in sequential correlation scenarios}

\author{Rodrigo Gallego}
\email[Electronic address:]{rgallego@zedat.fu-berlin.de}
\affiliation{Dahlem Center for Complex Quantum Systems, Freie Universität Berlin, 14195 Berlin, Germany}

\author{Lars Erik W\"urflinger}
\email[Electronic address:]{lars.wurflinger@icfo.es}
\affiliation{ICFO-Institut de Ciencies Fotoniques, Av. Carl Friedrich Gauss 3, E-08860 Castelldefels, Barcelona, Spain}

\author{Rafael Chaves}
\affiliation{Institute for Physics, University of Freiburg, Rheinstrasse 10, D-79104 Freiburg, Germany}

\author{Antonio Acín}
\affiliation{ICFO-Institut de Ciencies Fotoniques, Av. Carl Friedrich Gauss 3, E-08860 Castelldefels, Barcelona, Spain}
\affiliation{ICREA–Institució Catalana de Recerca i Estudis Avançats, Lluis Companys 23, 08010 Barcelona, Spain}

\author{Miguel Navascués}
\affiliation{H.H. Wills Physics Laboratory, University of Bristol,Tyndall Avenue, Bristol, BS8 1TL, United Kingdom}
\date{\today}

\begin{abstract}
As first shown by Popescu [S. Popescu, Phys. Rev. Lett. {\textbf{74}}, 2619 (1995)], some quantum states only reveal their nonlocality when subjected to a sequence of measurements while giving rise to local correlations in standard Bell tests. Motivated by this manifestation of ``hidden nonlocality" we set out to develop a general framework for the study of nonlocality when sequences of measurements are performed. Similar to [R. Gallego et al., Phys. Rev. Lett. {\textbf{109}}, 070401 (2013)] our approach is operational, i.e. the task is to identify the set of allowed operations in sequential correlation scenarios and define nonlocality as the resource that cannot be created by these operations. This leads to a characterisation of sequential nonlocality that contains as particular cases standard nonlocality and hidden nonlocality.
\end{abstract}

\maketitle
\section{Introduction}
One of the major problems in physics is to characterise
the different correlations that arise between distant observers
performing measurements on physical systems. The set of valid
correlations depends strongly on the theory that one uses to
describe such systems. Indeed,  such correlations allow one to
distinguish in an operational way classical theory, quantum theory
and more general probabilistic theories. For instance, since the
work of Bell \cite{Bell1964}, it has been known and widely studied
that some correlations obtained from measurements on quantum
systems cannot be simulated by any local and classical theory
(local hidden-variable models). This phenomenon is referred to as
nonlocality. It can be seen in an extremely simple scenario
consisting of two distant parties performing one out of two
possible measurements in each round of the experiment. Extensions
to more than two parties have also been studied giving rise to the
notion of multipartite nonlocality. Remarkably, nonlocality has
also been shown to be useful to perform quantum information tasks
without a classical analogue, such as quantum key distribution
\cite{Acin2007, Barrett2005, Acin2006} and random number generation
\cite{Pironio2010, Colbeck2011}.

In this work, we study a new form of correlations: those that
arise between distant observers performing a sequence of
measurements on their physical systems. Naively, one may think
that the sequence of measurements can be cast as an effective
unique measurement and that the scenario is not essentially
different from the standard one. However, we show that this
scenario is in many ways richer than the one with single
measurements per round. This is already implied by the results of
\cite{Popescu1995} where it was shown that some quantum states
display only local correlations in scenarios with single
projective measurements per round, but give rise to nonlocal
correlations when a sequence of measurements is performed instead.
This phenomenon was termed ``hidden nonlocality". Our main goal is to provide a general framework for the
study of correlations that arise from a sequence of measurements
on classical and quantum systems. Already in the classical case,
we show that different sets of correlations arise and hidden
nonlocality is only the tip of the iceberg. First, we study
classical hidden-variable models from the point of view of a
resource theory. We establish a set of operations that one can
perform on classical systems that do not create nonlocality
between distant observers. This set of operations is determined by
the causal structure established given by the sequence of
measurements performed by the observers. Relying on this
operational framework, we are able to define a notion of
``sequential nonlocality" that contains as  particular cases
standard nonlocality and hidden nonlocality. We further
investigate a type of hidden-variable models specially suited for
scenarios with sequential measurements, in the spirit of
time-ordered local models considered in
\cite{Gallego2011,Gallego2012}. These different notions of
classicality for the scenario with sequential measurements are
also generalized to the quantum case. This analysis leaves
numerous open problems that we enunciate and motivate for further
study.

\section{Sequential correlation scenarios}
In the study of correlation scenarios one typically assumes that
for every physical system, prepared by a common source, each party
chooses one measurement to perform and records the corresponding
result before the source generates a new physical system for the
next run of the correlation experiment. The data collected in this
way allows the parties when they come together after sufficiently
many runs to calculate the joint probabilities. This work studies
a generalisation of this situation where the parties are allowed
to perform a sequence of measurements on their part of the system
in every run of the experiment.
\subsection{General notation}
In general, these scenarios can be formulated for an arbitrary number of parties; for simplicity, however, we will restrict our study to the case of two parties. As in a usual correlation scenario with a single measurement per party in each round, a common source produces a bipartite physical system and sends one subsystem to $A$ and the other to $B$. In contrast to the standard situation, each party has now not only one set of possible measurement settings but one set of possible settings for each measurement of the sequence of measurements it is going to perform in each run of the experiment. To keep the notation simple let us start with the case of a sequence of two measurements for each party where we label the $i$-th measurement setting and the $i$-th outcome with $x_i,a_i$ and $y_i,b_i$ for $A$ and $B$ (figure \ref{fig:scenario}).
\begin{figure}[tbp]
\begin{center}
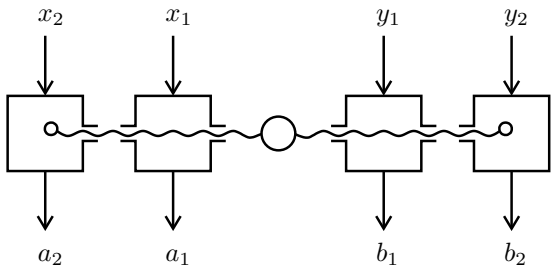
\end{center}
\caption{Sequential correlation scenario in the bipartite case. A common source prepares a physical system and each of the two parties receives a subsystem. The parties $A$ and $B$ choose the settings $x_1$ and $y_1$ for their first measurement respectively and obtain the outcomes $a_1$ and $b_1$; after recording the outputs of the first measurement they choose the measurements $x_2$ and $y_2$ yielding outcomes $a_2$ and $b_2$. The experiment is described by the collection of the joint probabilities $P(a_1a_2b_1b_2|x_1x_2y_1y_2)$.}
\label{fig:scenario}
\end{figure}
Then, the
observed correlations are the collection of joint probabilities
\begin{equation}
P(a_1 a_2 b_1 b_2 |x_1 x_2 y_1 y_2).
\end{equation}

Clearly, the outcome of the first measurement cannot depend on later measurement choices; but in the present scenario of sequential measurements later outcomes may well depend on the settings and outcomes of previous measurements. Therefore, the correlations $P(a_1 a_2 b_1 b_2| x_1 x_2 y_1 y_2)$ should not be viewed as four-partite nonsignalling correlations, but rather as a bipartite distribution, where no-signalling holds with respect to $A$ versus $B$ but where signalling from the first measurement of each party to the second of the same party is allowed.

Formally, the no-signalling condition between $A$ and $B$ means that a correlation $P$ that was obtained from a sequence of two measurements for each party obeys
\begin{align}
\label{eq:NScondition1}
 &\sum_{b_1,b_2} P(a_1a_2b_1 b_2 | x_1 x_2 y_1 y_2)  &&\text{is independent of } y_1,y_2,\\
\label{eq:NScondition2}
 &\sum_{a_1,a_2 } P(a_1a_2b_1 b_2 | x_1 x_2 y_1 y_2) && \text{is independent of } x_1,x_2,
\end{align}
which guarantees that the marginal distributions for $A$ and $B$
\begin{align}
P_A(a_1a_2|x_1x_2) &= \sum_{b_1,b_2}P(a_1a_2b_1 b_2 | x_1 x_2 y_1 y_2)\\
P_B(b_1b_2|y_1y_2) &= \sum_{a_1,a_2}P(a_1a_2b_1 b_2 | x_1 x_2 y_1 y_2)
\end{align}
are well defined. Furthermore, as later measurements cannot influence the outcome of previous ones, the correlations further have to fulfil
\begin{align}
\label{eq:NScondition3}
&\sum_{a_2}P(a_1a_2b_1b_2|x_1x_2y_1y_2)&& \text{independent of } x_2\\
\label{eq:NScondition4}
&\sum_{b_2}P(a_1a_2b_1b_2|x_1x_2y_1y_2)&& \text{independent of } y_2.
\end{align}

Correlations fulfilling the above conditions are the objects of interest in the study of sequential correlation scenarios. This notion can straightforwardly be generalised to the case of longer sequences of measurements.
%and more than two parties.

%\begin{dfn}[Sequential correlations]
%\label{dfn:seqscenario}
%For $n$ parties $A_1,\ldots,A_n$, where $A_i$ performs a sequence of $s_i$ measurements, let $(a_j^i$, $x_j^i)$ denote the $j$-th outcome and setting of the $i$-th party for $1\leq j\leq s_i$ and $1\leq i\leq n$. The correlations, given by the collection of the joint probabilities
%\begin{equation}
%P(\mathbf{a}^1\ldots \mathbf{a}^n|  \mathbf{x}^1\ldots \mathbf{x}^n)
%\end{equation}
%with $\mathbf{a}^i = (a^i_1,\ldots,a^i_{s_i}), \mathbf{x}^i = (x^i_1,\ldots,x^i_{s_i})$, are said to be \emph{$n$-partite sequential with respect to $\mathbf{s}=(s_1,\ldots, s_n)$}, if for  $1\leq i\leq n $
%\begin{align}
%\label{eq:seqconditions}
%&\sum_{a_j^i,\ldots, a^i_{s_i}} P(\mathbf{a}^1\ldots \mathbf{a}^n|  \mathbf{x}^1\ldots \mathbf{x}^n)&&\text{is independent of } (x_j^i,\ldots, x^i_{s_i})
%\end{align}
%for all $1\leq j \leq s_i$.
%\end{dfn}
\begin{dfn}[Sequential correlations]
\label{dfn:seqscenario2}
For two parties $A$ and $B$ that perform a sequence of $s$ and $t$ measurements respectively, let $(a_j$, $x_j)$ denote the $j$-th outcome and setting of $A$ for $1\leq j\leq s$ and $(b_j$, $y_j)$ the $j$-th outcome and setting of $B$ for $1\leq j\leq t$. The correlations, given by the collection of the joint probabilities
%\begin{equation}
%P(\mathbf{a}\mathbf{b}|\mathbf{x}\mathbf{y})
%\end{equation}
\begin{equation}
P(a_1\ldots a_{s}b_1\ldots b_{t}|x_1\ldots x_{s}y_1\ldots y_{t}),
\end{equation}
are said to be \emph{sequential with respect to $(s, t)$}, if
\begin{align}
\label{eq:seqconditionsa}
\sum_{a_j,\ldots, a_{s}}P(a_1\ldots a_{s}b_1\ldots b_{t}|x_1\ldots x_{s}y_1\ldots y_{t})
\end{align}
is independent of $(x_j,\ldots, x_s)$ for all $1\leq j \leq s$ and
\begin{align}
\label{eq:seqconditionsb}
\sum_{b_j,\ldots, b_{t}}P(a_1\ldots a_{s}b_1\ldots b_{t}|x_1\ldots x_{s}y_1\ldots y_{t})
\end{align}
is independent of $(y_j,\ldots, y_t)$ for all $1\leq j \leq t$.
\end{dfn}
\rodrigo{Notice that conditions \eqref{eq:seqconditionsa} and \eqref{eq:seqconditionsb}, when taking $j=1$ just express the no-signalling condition between parties $A$ and $B$. That is, the marginal statistics of $A$ ($B$) do not depend on the inputs chosen by $B$ ($A$). The remaining conditions obtained for $1<j\leq s$ in \eqref{eq:seqconditionsa} --or $1<j\leq t$ in \eqref{eq:seqconditionsb}-- capture the temporal ordering between the sequence of measurements chosen by $A$ and $B$. For instance, taking $j=2$ in \eqref{eq:seqconditionsa} we arrive at the condition that $P(a_1,b_1,\ldots,b_t|x_1,\ldots,x_s,y_1,\ldots,y_t)$ is independent of $x_2,...,x_s$. This must be the case as the inputs $x_2,...,x_s$ should not influence neither $B$ nor the previous event $(a_1,x_1)$ in $A$.}

%Interpreting the output-input pairs $(a_1\ldots a_{s}|x_1\ldots x_{s})$ and $(b_1\ldots b_{t}|y_1\ldots y_{t})$ as the overall output-input pairs of parties $A$ and $B$, sequential correlations as in the above definition constitute a bipartite nonsignalling box. However, the conditional probabilities expressed  in the input-output pairs of the individual measurements $\lbrace (a_i,x_i) |1\leq i \leq s\rbrace$ and $\lbrace (b_i,y_i) |1\leq i \leq t\rbrace$ in general need not fulfil the no-signalling principle, if seen as the collection of the probabilities of the $s+t$ variables $\lbrace a_i,b_j\rbrace$ conditioned on the $s+t$ variables $\lbrace x_i,y_j \rbrace$. The signalling displayed by the distribution when expressed in the individual inputs and outputs is constrained by equations \eqref{eq:seqconditionsa} and \eqref{eq:seqconditionsb}, which allow signalling only among the measurements of the same party in such a way that outcomes can depend at most on all the settings chosen by that party so far.

Ignoring the length of the input and output alphabets for each measurement, a sequential correlation scenario is then characterised by $(s,t)$ specifying the length of the measurement sequence for each party. %Therefore, we also refer to correlations as in \ref{dfn:seqscenario} as $n$-partite correlations with respect to the sequence $\mathbf{s}$.

%Given $n$ parties $A_1,\ldots, A_n$, where $A_i$ performs a sequence of $s_i$ measurements, the correlations are given by
%\begin{equation}
%\label{eq:sequence}
%P(\mathbf{a}^1\ldots \mathbf{a}^n|  \mathbf{x}^1\ldots \mathbf{x}^n)
%\end{equation}
%with $\mathbf{a}^i = (a^i_1,\ldots,a^i_{s_i}),       \mathbf{x}^i = (x^i_1,\ldots,x^i_{s_i})$ and $a_j^i$, $x_j^i$ denote the $j$-th outcome and setting of the $i$-th party for $1\leq j\leq s_i$ and $1\leq i\leq n$. The conditions corresponding to \ref{eq:NScondition1,eq:NScondition2,eq:NScondition3,eq:NScondition4} can then be summarised as
%\begin{align}
%&\sum_{a_j^i,\ldots, a^i_{s_i}} P(\mathbf{a}^1\ldots \mathbf{a}^n|  \mathbf{x}^1\ldots \mathbf{x}^n)&&\text{independent of } (x_j^i,\ldots, x^i_{s_i})
%\end{align}
%for all $1\leq i\leq n $ and $1\leq j \leq s_i -1$.

\subsection{Known notions of nonlocality}

Before we turn to the study of nonlocality in scenarios with sequential measurements let us review  known notions of nonlocality, namely standard Bell nonlocality \cite{Bell1964} and `hidden nonlocality' first considered by Popescu \cite{Popescu1995}.

%Standard Bell nonlocality

The standard notion of locality is due to Bell, who showed that correlations arising from local measurements on a bipartite quantum system cannot always be explained by a local hidden-variable model \cite{Bell1964}. More precisely, consider the situation of two parties $A$ and $B$ performing local measurements on their share of a bipartite system. Each system can just be seen as a black box that accepts a classical input, $x$ for $A$ and $y$ for $B$, and produces a classical output, $a$ and $b$ respectively. The correlations between the input and output processes among the two parties are then described by the joint conditional probability distribution $P(ab|xy)$. The no-signalling principle for this case states that
\begin{equation}
 \label{eq:nscond}
\begin{aligned}
 \sum_b P(ab|xy) &\quad \text{is independent of } y\\
\sum_a P(ab|xy) &\quad \text{is independent of } x.
\end{aligned}
\end{equation}

Correlations are said to admit a local hidden-variable model if the probability distribution can be decomposed as
\begin{equation}
 \label{eq:bellLocal}
P(ab|xy) = \sum_\lambda p_\lambda P_A^\lambda (a| x ) P_B^\lambda (b|y ),
\end{equation}
where $p_\lambda$ is a probability distribution of the hidden variable $\lambda$ and $P_A^\lambda, P_B^\lambda$ are probability distributions of the outcomes $a,b$ given the inputs $x,y$ respectively. %If a model as in \eqref{eq:bellLocal} exists, then it is always possible to find a model where the local maps are deterministic \cite{Fine1982a}.
\rodrigo{There exist different equivalent formulations of the set of assumptions leading to \eqref{eq:bellLocal}. For instance, the assumption of local-realism --$P^{\lambda}(ab|xy)=P_A^\lambda (a| x ) P_B^\lambda (b|y )$-- together with the so-called assumption of measurement independence --stating that $p_{\lambda}(xy)=p_{\lambda}$-- lead straightforwardly to \eqref{eq:bellLocal}. For sake of conciseness we incur in an abuse of terminology denoting correlations of the form \eqref{eq:bellLocal} simply as local. Equivalently, correlations that cannot be written as \eqref{eq:bellLocal} will be termed nonlocal.}

The set of all probability distributions that fulfil \eqref{eq:nscond} and \eqref{eq:bellLocal} define the convex polytope of local correlations \cite{Pitowsky1989}. Any no-signalling probability distribution that does not admit a local hidden-variable model violates at least one of the facet defining inequalities of the local polytope; these inequalities are called Bell inequalities and it is known that every pure entangled quantum state violates a Bell inequality \cite{Gisin1991}.

%Hidden nonlocality

While every pure entangled state violates some Bell inequality, there are entangled mixed states that do not violate any Bell inequality as shown by Werner \cite{Werner1989}.
He introduced a class of bipartite mixed states that, for a certain parameter region, are entangled but notwithstanding admit a local hidden-variable model for all possible projective measurements. These \emph{Werner states} $W$ act on $\mathbb{C}^d \otimes \mathbb{C}^d$ and are of the form
\begin{equation}
 \label{eq:WernerStates}
W = p \frac{\identity + \mathbb{F}}{d(d+1)} + (1-p) \frac{\identity -\mathbb{F} }{d(d-1)},
\end{equation}
where $\identity$ denotes the identity matrix on the $d\times d$ dimensional Hilbert space, $\mathbb{F}=\sum_{ij}\ket{i}\bra{j}\otimes \ket{j}\bra{i}$ the flip operator, and $0\leq p \leq 1$. For $p=\frac{1+d}{2d^2}$ these states are entangled but do not violate any standard Bell inequality. Werner proved this by constructing an explicit local hidden-variable model that reproduces the correlations of these states for arbitrary projective measurements \citep{Werner1989}.

However, as noted by Popescu in \cite{Popescu1995}, a state like \eqref{eq:WernerStates} can give rise to correlations that are incompatible with an explanation by local hidden-variables if it is subjected to a sequence of measurements, thereby revealing what he named ``hidden'' nonlocality. To see how this argument works, suppose that the system is subjected to a sequence of two projective measurements for each party. First, each party performs a measurement that corresponds to the projection of that party's subsystem onto a two-dimensional subspace (or its orthogonal complement), i.e. $A$ performs the measurement given by the projections $\lbrace P,\identity_d -P\rbrace$ and $B$ the measurement given by the projections $\lbrace Q,\identity_d-Q\rbrace$, where
\begin{align}
 P&= \proj{1}_A+\proj{2}_A\\
 Q&= \proj{1}_B+\proj{2}_B.
\end{align}
Now, after recording the outcome of the first measurements the parties choose their measurement settings for the second round of measurements. When the parties obtained the outcomes corresponding to the projections $P$ and $Q$ respectively in the first round, the post-measurement state is given by
\begin{align}
 W' &= \frac{P\otimes Q W P \otimes Q}{\trace (P\otimes Q W)}\\
&=\frac{d}{d +2 }\left(\frac{1}{2d}\identity_4 + \proj{\Psi^-}\right),
\end{align}
where
\begin{equation}
\ket{\Psi^-}=\frac{1}{\sqrt{2}}(\ket{1} \ket{2}-\ket{2} \ket{1})
\end{equation}
the singlet state. If the parties now choose observables
$A_0,A_1,B_0,B_1$ for $A$ and $B$ respectively that give the
maximal value of the Clauser-Horne-Shimony-Holt (CHSH) expression
\cite{Clauser1969} for the singlet state, they obtain the
following value
\begin{align}
\label{eq:chshviol}
 \beta =\trace(W'(A_0B_0 +A_0B_1+A_1B_0-A_1B_1)) =& \frac{2 \sqrt{2} d}{d+2}
\end{align}
that depends on the local dimension $d$. For $d\geq 5$ we have $\beta > 2$, a violation of the CHSH inequality, indicating that in this case the observed correlations cannot be explained by a local hidden-variable model. So, even though a local hidden-variable model can account for all the correlations obtained from Werner states with $d\geq 5$ that result from a single round of local projective measurements, no such model could account for the correlations obtained when a sequence of two projective measurements is performed by each party.

The example \cite{Popescu1995} gave was concerned with a specific class of quantum states. Clearly, those states display some sort of nonlocality but they have a standard local hidden-variable model. Thus, the question naturally arises how to formulate the notion of locality in sequential correlation scenarios.

Another example of hidden nonlocality was reported in \cite{Gisin1996} where examples of entangled states in dimension $d=2$ were presented that do not violate the CHSH inequalities for rounds of single measurements but do so when sequences of generalised measurements, given by positive-operator valued measures (POVMs), are performed.
It is worth mentioning that, until recently, all known examples of hidden non-locality were derived for states that are local under projective measurements. Therefore, it was open whether the non-locality of these states could simply be detected by allowing general measurements. In fact, while there exists a local model for general measurements acting on some entangled Werner states \cite{Barrett2002a}, these states do not display hidden non-locality \`a la Popescu. This  question was solved in a recent work, where a class of entangled two-qubit states was presented that have a local model for general measurements (POVMs) for rounds of single measurements, but violate a Bell inequality after local filtering operations \cite{Hirsch2013}.

For the sake of simplicity let us for now focus on the case considered by Popescu, i.e. a sequence of two measurements for each party, where the first measurement by each party is always the same. We will denote by $x_2,y_2$ the measurement settings for the second measurement and by $a_i,b_i$ for $i=1,2$ the outcome of the $i$-th measurement of the parties $A,B$.

Obviously, the notion of locality in the current scenario of sequential measurements should include the standard notion of locality in the sense of Bell, that the probability distribution  should be able to be decomposed as
\begin{equation}
 \label{eq:HiddenNonlocality1}
P(a_1a_2b_1 b_2 | x_2  y_2)= \sum_\lambda p_\lambda P_A^\lambda(a_1a_2| x_2 ) P_B^\lambda(b_1b_2|y_2).
\end{equation}
After the discussion of Popescu's example it is also clear that another necessary condition for an appropriate definition of locality in sequential scenarios is the absence of hidden nonlocality. Therefore, one will further require that all possible post-selections have a local model as well, i.e.
\begin{align}
 \label{eq:HiddenNonlocality2b}
P(a_1a_2 b_2 | x_2 y_2  b_1 )&= \sum_\lambda p_{\lambda|  b_1} P_A^\lambda(a_1 a_2| x_2 ) P_B^\lambda(b_2| y_2)\\
\label{eq:HiddenNonlocality2a}
P(a_2 b_1 b_2 | x_2 y_2  a_1 )&= \sum_\lambda p_{\lambda|  a_1} P_A^\lambda(a_2| x_2 ) P_B^\lambda(b_1b_2| y_2)\\
\label{eq:HiddenNonlocality2ab}
P(a_2 b_2 | x_2 y_2 a_1  b_1 )&= \sum_\lambda p_{\lambda| a_1 b_1} P_A^\lambda(a_2| x_2 ) P_B^\lambda(b_2| y_2)
\end{align}
for all values of $(a_1,b_1)$ and where the weights $p_{\lambda|\cdot}$ will in general depend on the outputs of the first measurement round.

Let us return to the explicit example by Popescu. Denote the first measurements of $A$ and $B$ by the projectors $P_{a_1}$ and $Q_{b_1}$; and the second measurements by $\tilde{P}_{a_2|x_2}$ and $\tilde{Q}_{b_2|y_2}$. Then the probabilities read
\begin{multline}
\label{eq:popescuqu}
P(a_1a_2b_1 b_2 | x_2  y_2)\\= \trace\left(P_{a_1}\tilde{P}_{a_2|x_2}P_{a_1}\otimes Q_{b_1}\tilde{Q}_{b_2|y_2}Q_{b_1} W\right).
\end{multline}

Now, for Popescu's example the projections of the first and second measurement commute for both $A$ and $B$. Thus, the expression \eqref{eq:popescuqu} can be seen as correlations arising from a single projective measurement on each side and are therefore, due to the explicit hidden-variable model constructed by Werner, local in the sense of \eqref{eq:HiddenNonlocality1}. On the other hand,  they do not fulfil the condition \eqref{eq:HiddenNonlocality2ab}, for the probabilities post-selected on the first outcome of the first measurement,
\begin{equation}
 P(a_2 b_2 | x_2 y_2 , a_1=1, b_1 =1 ),
\end{equation}
violate the CHSH inequality.

As said, the first condition \eqref{eq:HiddenNonlocality1} is nothing but the standard locality condition in the spirit of Bell between the two parties $A$ and $B$ when the pairs $(a_1,a_2)$ and $(b_1,b_2)$ are seen as one output for $A$ and $B$ respectively. The remaining conditions \eqref{eq:HiddenNonlocality2b} through \eqref{eq:HiddenNonlocality2ab} ensure that there is no hidden nonlocality as the correlations that arise from the subsequent measurement can be simulated by a local hidden-variable model whatever results were obtained in the first measurement round. As we will see in the following, these necessary requirements are in general not sufficient to capture the notion of locality in a sequential correlation scenario.

\section{Nonlocality in sequential correlation scenarios}

Popescu's example already showed that the standard notion of locality is not sufficient to capture the behaviour of correlations that can arise in a sequential correlation scenario. As said, the correlations obtained from rounds of single measurements are local as they admit a standard local hidden-variable model. One way to interpret the nonlocality revealed in Popescu's example is to observe that the correlations do not fulfil condition \eqref{eq:HiddenNonlocality2ab}. On the other hand, the post-selection of the ensemble in \eqref{eq:HiddenNonlocality2ab} is an operation that can be performed by the two parties locally and can thus be seen as a local preparation of a physical resource that is then subjected to a standard Bell test. This way of looking at Popescu's example gives rise to an operational definition of nonlocality. In the following we present a general framework for this operational characterisation of nonlocality adapted to sequential correlation scenarios.

\subsection{Time-ordered local models}
Before we turn to the task of characterising nonlocality in operational terms, however, let us examine the lack of a decomposition as in \eqref{eq:HiddenNonlocality2ab} for sequential correlations. This form of nonlocality of correlations can be understood as not admitting a local and causal hidden-variable model as mentioned in \cite{Zukowski1998}.

\begin{dfn}[Time-ordered local models]
\label{dfn:sequentiallocal}
Given sequential correlations with respect to $(s,t)$, described by
the joint probabilities
\begin{equation}
P(a_1\ldots a_s b_1 \ldots b_t|x_1\ldots x_s y_1 \ldots y_t).
\end{equation}
The set $\mathsf{TOLoc}$ of correlations admitting a \emph{time-ordered local model} is given by correlations that can be decomposed as
\begin{multline}
\label{eq:TOLocaldecomp}
P(a_1\ldots a_s b_1 \ldots b_t|x_1\ldots x_s y_1 \ldots y_t)\\
=\sum_\lambda p_\lambda P_A^\lambda(a_1\ldots a_s|x_1\ldots x_s)P_B^\lambda(b_1\ldots b_t|y_1\ldots y_t)
\end{multline}
where the positive weights $p_\lambda$ sum to unity and the conditional probabilities $P_A^\lambda$ and $P_B^\lambda$ are sequential, i.e. for all $\lambda$ we have
\begin{align}
\label{eq:causalmodela}
\sum_{a_j,\ldots, a_{s}}P_A^\lambda(a_1\ldots a_{s}|x_1\ldots x_{s})
\end{align}
is independent of $(x_j,\ldots, x_s)$ for all $1\leq j \leq s$ and
\begin{align}
\label{eq:causalmodelb}
\sum_{b_j,\ldots, b_{t}}P_B^\lambda(b_1\ldots b_{t}|y_1\ldots y_{t})
\end{align}
is independent of $(y_j,\ldots, y_t)$ for all $1\leq j \leq t$.
\end{dfn}
\rodrigo{Let us illustrate the definition of Time-ordered local models for a simple example of $s=t=2$. In this case, a time-ordered local model reads
\begin{multline}
P(a_1 a_2 b_1 b_2|x_1 x_2 y_1 y_2)\\
=\sum_\lambda p_\lambda P_A^\lambda(a_1 a_2|x_1 x_2)P_B^\lambda(b_1 b_2|y_1y_2)
\end{multline}
where we demand also that $\sum_{a_2}P^{\lambda}_A(a_1a_2|x_1x_2)$ is independent of $x_2$ for all $\lambda$, and equivalently for $P_B^{\lambda}$. This should be interpreted as $\lambda$ carrying information about the local instructions $P_A^{\lambda}$ and $P_B^{\lambda}$, which fulfil the condition imposed by the sequential ordering of the measurements: the behaviour of $(x_1,a_1)$ does not depend on the posterior input $x_2$, and equivalently for Bob. }

Let us compare these models with the standard  formulation of local hidden-variable models by Bell. The theorem of Bell assumes a certain causal structure between the hidden variable $\lambda$ and the events of measurements $x,y$ and outcomes $a,b$ of two separated parties to derive linear constraints, in the form of inequalities, on the joint probabilities $P(ab|xy)$.

Formally, a \emph{causal structure} is a set of variables $V$ and a set of ordered pairs of distinct variables $(x,a)$ determining that $x$ is a direct cause of $a$ relative to $V$ \citep{Pearl2009,Spirtes2001}. A convenient way to represent causal structures is through \emph{directed acyclic graphs} (DAGs), where every variable $x\in V$ is a vertex and every ordered pair $(x,a)$ is represented by a directed edge from $x$ to $a$.

In the standard Bell scenario of two parties there are the observed variables $x,y,a,b$ and further the hidden variable $\lambda$, a common cause of both outputs $a$ and $b$. Thus, in this case one arrives at the causal structure presented in figure \ref{fig:causalstruct}.

\begin{figure}[tbp]
\begin{center}
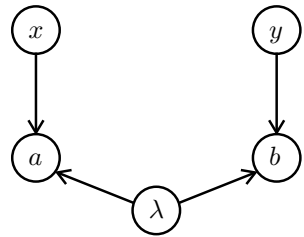
\end{center}
\caption{Causal structure of the standard bipartite Bell scenario. The observed variables are the inputs $x,y$ of the two parties and their respective outputs $a$ and $b$; further a hidden variable $\lambda$ is assumed as a common cause for both $a$ and $b$.}
\label{fig:causalstruct}
\end{figure}

As for time-ordered local models let us for definiteness start
with the simple case of two parties each performing a sequence of
two measurements. The observed variables in this case are
$x_1,x_2,y_1,y_2,a_1,a_2,b_1,b_2$, where $x_i$ and $a_i$ denote
the $i$-th measurement setting and $i$-th outcome for $A$, and
$y_i$ and $b_i$ denote the $i$-th measurement setting and $i$-th
outcome for $B$; further a hidden variable $\lambda$ that is a
common cause for all outputs is assumed. Clearly, there are direct
causal influences from $x_i$ to $a_i$, from $y_i$ to $b_i$ and
from $\lambda$ to all outputs. As we are treating the parties $A$
and $B$ as separated, we exclude causal influences from inputs of
one party to the outputs of the other. Later measurement outcomes
of one party, however, will in general depend on earlier settings
or outcomes of that party. The response of one party for its
$i$-th measurement should depend only  on the given hidden
variable $\lambda$, the first $i$ measurement settings and the
first $i-1$ measurement outcomes of that party.  The resulting
causal structure is shown in figure \ref{fig:seqstructure}.

%\begin{dfn}[Local-causal models]
%\label{dfn:sequentiallocal}
%Let $P$ be an $n$-partite sequential correlation with respect to $\mathbf{s}$ as in \ref{dfn:seqscenario}, described by
%the joint probabilities $P(\mathbf{a}^1\ldots \mathbf{a}^n|  \mathbf{x}^1\ldots \mathbf{x}^n)$.
%%,where %$\mathbf{a}^i = (a^i_1,\ldots,a^i_{s_i}), \mathbf{x}^i = (x^i_1,\ldots,x^i_{s_i})$ for $1\leq i \leq n$.
%$P$ is said to have a \emph{local-causal model}, if the probabilities can be decomposed as
%\begin{equation}
%\label{eq:TOLocaldecomp}
%P(\mathbf{a}^1\ldots \mathbf{a}^n|  \mathbf{x}^1\ldots \mathbf{x}^n) = \sum_\lambda p_\lambda P^\lambda_1(\mathbf{a}^1|\mathbf{x}^1)\ldots P^\lambda_n(\mathbf{a}^n|\mathbf{x}^n),
%\end{equation}
%where the positive weights $p_\lambda$ sum to unity and the conditional probabilities $P_i^\lambda$ are sequential, i.e. for all $\lambda$ and all $i$ we have
%\begin{align}
%\label{eq:TOLocal}
%&\sum_{a_j^i,\ldots, a^i_{s_i}} P_i^\lambda(a^i_1\ldots a_{s_i}^i|  x^i_1\ldots x_{s_i}^i)&&\text{independent of } (x_j^i,\ldots, x^i_{s_i})
%\end{align}
%for $1\leq j \leq s_i$.
%\end{dfn}
Note, any collection of conditional probabilities $\lbrace
P_A^\lambda, P_B^\lambda| \lambda \in \Lambda \rbrace$ fulfilling
the conditions of \eqref{eq:causalmodela} and
\eqref{eq:causalmodelb} defines via \eqref{eq:TOLocaldecomp} valid
bipartite sequential correlations admitting a time-ordered local
model for any distribution $p_\lambda$ of the hidden variable.
Once we fix the causal structure, expressed in the conditions
\eqref{eq:causalmodela} and \eqref{eq:causalmodelb}, no
fine-tuning of the model parameter $p_\lambda$ is needed to
obtain correlations with the correct causal independence
conditions. The fact that the models defined on the causal
structure shown in figure \ref{fig:seqstructure} do not require
fine-tuning makes them the natural generalisation of local
hidden-variable models to sequential correlation scenarios. As
already noted in \cite{Zukowski1998}, one can easily see that
correlations admitting such a model do not display hidden
nonlocality.

\begin{figure}[tbp]
\begin{center}
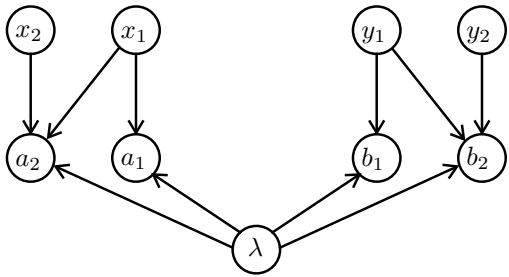
\end{center}
\caption{Causal structure for the bipartite sequential correlation scenario with sequences of two measurements for each party. The observed variables are the inputs $x_1,x_2$ of the first party, the inputs $y_1, y_2$ of the second party and the corresponding outputs $a_1,a_2$ and $b_1,b_2$. The first output of one party is determined by the first input of that party and the hidden variable $\lambda$; the second output depends on both inputs of the respective party and the hidden variable $\lambda$.}
\label{fig:seqstructure}
\end{figure}

\begin{prop}
\label{prop:nohidden}
Let $P$ be sequential correlations with respect to $(s,t)$ that admit a time-ordered local model.
Then all correlations obtained by post-selection admit a time-ordered local model.
\end{prop}

\begin{proof}
Consider post-selection on the first measurement of $A$ to yield $a_1$ given the setting $x_1$. Given the time-ordered local model for $P$
\begin{multline}
P(a_1\ldots a_s b_1 \ldots b_t|x_1\ldots x_s y_1 \ldots y_t)\\
=\sum_\lambda p_\lambda P_A^\lambda(a_1\ldots a_s|x_1\ldots x_s)P_B^\lambda(b_1\ldots b_t|y_1\ldots y_t)
\end{multline}
the post-selected correlations are given by
\begin{multline}
P_{a_1|x_1}(a_2\ldots a_s b_1 \ldots b_t|x_2\ldots x_s y_1 \ldots y_t)\\
=\frac{1}{P(a_1|x_1)}P(a_1\ldots a_s b_1 \ldots b_t|x_1\ldots x_s y_1 \ldots y_t)
\end{multline}
and can be written as
\begin{multline}
P_{a_1|x_1}(a_2\ldots a_s b_1 \ldots b_t|x_2\ldots x_s y_1 \ldots y_t)\\
= \sum_\lambda \tilde{p}_\lambda \tilde{P}_A^\lambda(a_2\ldots a_s|x_2\ldots x_s)P_B^\lambda(b_1\ldots b_t|y_1\ldots y_t),
\end{multline}
where we defined
\begin{align}
\tilde{p}_\lambda =\frac{p_\lambda P_A^\lambda(a_1|x_1)}{P_A(a_1|x_1)}
\end{align}
and
\begin{align}
\tilde{P}^\lambda_A(a_2\ldots a_s|x_2\ldots x_s)=\frac{P_A^\lambda(a_1\ldots a_s|x_1\ldots x_s)}{P_A^\lambda(a_1|x_1)}.
\end{align}
As the resulting correlations admit a time-ordered local model, further post-selections by either party will again result in correlations admitting such a model.
\end{proof}

Thus, one way to understand the nonlocality revealed by Popescu's argument is to observe that the corresponding correlations do not admit a time-ordered local model as in definition \ref{dfn:sequentialLocality}.
\subsection{Operational characterisation of nonlocality}
As said above, however, identifying post-selection as a possible operation that can be performed locally by the parties opens the possibility to characterise nonlocality operationally. So, instead of defining locality in a given correlation scenario through a specific class of models we set out to define nonlocality in operational terms; this was done in \cite{Gallego2012} for the case of multipartite nonlocality with single measurements in each round. The general idea within this approach is to interpret nonlocality as a resource, i.e a property of correlations that cannot be created by a certain set of allowed operations. Just as one can define entangled states as non-separable states, one can alternatively define entanglement as the resource that cannot be created between two distant laboratories by the use of local operations and classical communication (LOCC) \cite{Horodecki2009}. In \cite{Gallego2012} an analogous operational framework for the resource of nonlocality was developed. In the following we briefly review the elements of this framework and adapt them to the current scenario of sequential measurements.

The set of objects in the present scenario is the set of sequential correlations described by joint probabilities of the form
\begin{equation}
P(a_1\ldots a_{s}b_1\ldots b_{t}|x_1\ldots x_{s}y_1\ldots y_{t})
\end{equation}
that, as in definition \ref{dfn:seqscenario2}, can be obtained in a correlation experiment with sequences of measurements. To characterise nonlocality as a resource we need to specify the set of allowed operations in this scenario. A valid protocol for the two parties consists of two stages, the \emph{preparation stage} and the \emph{measurement stage}.

The first stage takes place before the inputs for the final
nonlocality test are provided and allows the parties to perform
measurements on their share of the physical system and communicate
the corresponding results among each other. Depending on the
obtained and communicated results the parties may choose to
perform further measurements.

Note that classical communication is allowed at this stage;  as
the parties have not yet received the inputs for the final Bell
test, however, this communication cannot be used to create any
nonlocal correlations.

The second stage starts when the inputs for the final nonlocality experiment are provided; from this point on no more communication is allowed. The local operations every party can now perform consist of processing the classical inputs and outputs, referred to as \emph{wirings}. Upon receiving the input $x$ party $A$ can choose an arbitrary function $f_1$ to determine the input $x_1$ for the first measurement yielding an outcome $a_1$; the inputs for the following measurements are determined by a function of the provided input $x$ and all measurement outcomes obtained so far, i.e. $x_i = f_i(x,a_1, \ldots, a_{i-1})$. Lastly, the final output $a$ is determined by a function $g$ of the input $x$ and all outputs $a_i$. Or, more formally,
\begin{dfn}[Sequential wiring]
\label{dfn:seqwiring}
Let $P$ be bipartite sequential correlations with respect to  $(s,t)$. A \emph{sequential wiring} for party $A$ is specified by functions $\lbrace f_1,\ldots f_s,g\rbrace$ and takes $P$ to the correlations $P'$, where $P'$ is characterised by
\begin{multline}
P'(a b_1\ldots b_t|x y_1\ldots y_t)\\
=\sum_{\substack{a_1\ldots a_s\\ \mathrm{s.t.}\;g(x,a_1\ldots a_s)=a}} P(a_1\ldots a_s b_1 \ldots b_t|\xi_1\ldots \xi_s y_1\ldots y_t)
\end{multline}
with $\xi_1 = f_1(x)$ and $\xi_i = f_i(x, a_1,\ldots, a_{i-1})$ for $i\geq 2$.
\end{dfn}

It is easy to see that the resulting correlations $P'$ are well defined; in particular they are sequential with respect to $(1,t)$. Analogously, one can define a wiring for party $B$. Furthermore, a sequential wiring can act on $n$ different sequential probability distributions members. In this case, a set of extra functions specifies the order according to which each party measures its share of the $n$ systems, which may depend on the input  and the previous outcomes.

With the allowed operations defined, we can now define nonlocality as the property of sequential correlations that cannot be created by these allowed operations.

\begin{dfn}[Operationally local correlations]
\label{dfn:sequentialLocality}
Consider the set of bipartite correlations that are sequential with respect to $(s,t)$. Sequential correlations $P$ belong to the set $\mathsf{OpLoc}$ of \emph{operationally local correlations}, if for any $n$ the product correlations $P^{\times n}$ are mapped  by any valid protocol to bipartite correlations that are local in the standard sense.
\end{dfn}

\rodrigo{In other words, correlations $P \in \mathsf{OpLoc}$ cannot be employed to violate a standard Bell inequality between $A$ and $B$. Even if we allow for an arbitrary large number of copies of $P$ and any local processing of information.}
In the standard bipartite scenario with single measurements the operational definition of locality coincides with the characterisation by local hidden-variable models as in \eqref{eq:bellLocal} \cite{Gallego2012}. For the present situation of sequential correlations we can show that correlations admitting a time-ordered local model are compatible with our operational definition in the sense that no allowed operation can map time-ordered local correlations into nonlocal correlations.

\begin{prop}
\label{thm:modelscomp}
The set of time-ordered local correlations is contained in the set of operationally local correlations, i.e. $\mathsf{TOLoc} \subset \mathsf{OpLoc}$.
\end{prop}
\begin{proof}

%We need to show that any valid protocol that takes an arbitrary number of distributions each admitting a time-ordered local model as input results in a final bipartite distribution $P_\mathrm{fin}$ admitting a local decomposition, i.e.
%\begin{equation}\label{eq:lhvmfin} P_\mathrm{fin}(ab|xy)=\sum_\lambda p_\lambda P_A^\lambda(a|x)P^\lambda_B(b|y).
%\end{equation}
%Now, by assumption, for $1\leq i \leq n$ let $P^{(i)}$ be sequential correlations that admit a time-ordered local model, i.e.
%\begin{multline}\label{eq:timeorderedn}
%P^{(i)}(a_1^i\ldots a_l^i b^i_1 \ldots b^i_r|a_1^i\ldots a_l^i b^i_1 \ldots b^i_r)=\\
%\sum_{\lambda_i} p^{(i)}_{\lambda_i} P_{A_i}^{\lambda_i}(a^i_1\ldots a^i_l|x_1^i\ldots x_l^i) P_{B_i}^{\lambda_i}(b^i_1\ldots b^i_r|y_1^i\ldots y_r^i)
%\end{multline}
%that are distributed between $A$ and $B$.

During the preparation stage $A$ and $B$ implement measurements on a subset of the $n$ systems shared by them and communicate the corresponding results. The resulting correlations are nothing but post-selections of the original correlations, which by proposition \ref{prop:nohidden} admit a time-ordered local model.

The local sequential wirings applied during the measurement stage map these time-ordered local correlations to bipartite local correlations.
\end{proof}

\rodrigo{Another way of formulating Prop. \ref{thm:modelscomp} is the following: given sequential correlations $P$, it suffices to find a decomposition as in \eqref{eq:TOLocaldecomp} to ensure that any protocol processing locally $n$ copies of $P$ will not violate a Bell inequality between $A$ and $B$.}
The set of correlations admitting a time-ordered local model is therefore contained in the set of correlations that are local in the operational sense. From this operational point of view we see that in Popescu's example the parties use the operation of post-selection to create nonlocal bipartite correlations.

However, the situation described by Popescu is special insofar as there is only one setting for the first measurements of $A$ and $B$. As we show next, in a situation like this the existence of a time-ordered local model for the correlations is equivalent to all post-selections having a time-ordered local model.

\begin{prop}
Consider sequential correlations with respect to $(s,t)$ and assume that for the first measurement by $A$ and the first measurement by $B$ there is only one setting each, i.e. the correlations are characterised by
\begin{equation}
P(a_1\ldots a_s b_1 \ldots b_t|x_{2}\ldots x_s y_{2}\ldots y_t).
\end{equation}
Then the following are equivalent:
\begin{itemize}
\item[(i)] P admits a time-ordered local model with respect to $(s,t)$.
\item[(ii)]All post-selections on the first measurement of $A$ and the first measurement of $B$ admit a time-ordered local model with respect to $(s-1,t-1)$.
\end{itemize}
\end{prop}

\begin{proof}
That (i) implies (ii) is clear by proposition \ref{prop:nohidden}. To see the converse consider the post-selections
\begin{multline}
P(a_{2}\ldots a_s b_{2}\ldots b_t|x_{2}\ldots x_s y_{2}\ldots y_t, a_1 b_1)=\\
\sum_\lambda p_{\lambda|a_1,b_1} P^\lambda_A(a_2 \ldots a_s|x_2 \ldots x_s)P^\lambda_B(b_2 \ldots b_t|y_2 \ldots y_t)
\end{multline}
all of which admit a time-ordered local model by assumption. Define a new hidden variable $\mu = (\lambda, a_1', b_1')$ distributed according to $q_\mu = p_{\lambda |a_1',b_1'} P(a_1'b_1')$ and sequential response functions
\begin{align}
\tilde{P}^\mu_A(a_1\ldots a_s|x_1\ldots x_s) &= \delta_{a_1}^{a_1'} P_A^\lambda(a_2 \ldots a_s|x_2 \ldots x_s)\\
\tilde{P}^\mu_B(b_1\ldots b_t|y_1\ldots y_t) &= \delta_{b_1}^{b_1'} P_B^\lambda(b_2 \ldots b_t|y_2 \ldots y_t)
\end{align}
to get the time-ordered local model with respect to $(s,t)$
\begin{multline}
P(a_{1}\ldots a_s b_{1}\ldots b_t|x_{1}\ldots x_s y_{1}\ldots y_t)\\
=\sum_\mu q_\mu \tilde{P}^\mu_A(a_1 \ldots a_s|x_1 \ldots x_s)\tilde{P}^\mu_B(b_1 \ldots b_t|y_1 \ldots y_t).
\end{multline}
\end{proof}

So for scenarios where the first measurements have only one setting, i.e. the first measurements are always the same, the absence of hidden nonlocality when post-selecting on these first measurements is equivalent to the existence of a time-ordered local model for the correlations. In particular, in a scenario as described by Popescu we have that all sequential correlations whose post-selections fulfil the CHSH inequality necessarily admit a time-ordered local model. In the general case, however, the question arises whether there are forms of nonlocality that are different from standard Bell nonlocality or Popescu's hidden nonlocality.

In other words, are there correlations that are local in the standard notion, do not display hidden nonlocality, but nevertheless need to be considered nonlocal in the operational sense?

\subsection{A simple scenario}
\label{sec:detection}

To give a first answer to this question we will consider the simplest non-trivial case of sequential measurements in a bipartite scenario, namely one measurement for party $A$ and a sequence of two measurements for $B$, where for each measurement the respective party can choose from two settings yielding one of two possible outcomes. We denote the outcomes of $A$ and $B$ by $a, b_1, b_2 \in \lbrace +1, -1\rbrace $ and the inputs by $x,y_1,y_2 \in \lbrace 0,1 \rbrace$ and consider the joint probabilities $P(a b_1 b_2 |x y_1 y_2)$.

Thus, in this scenario $\mathsf{TOLoc}$ denotes the set of sequential correlations $P$ that admit a time-ordered local model for the given scenario, i.e. for $P \in \mathsf{TOLoc}$ we have
\begin{equation}
 P(a b_1 b_2 |x y_1 y_2)= \sum_\lambda p_\lambda P_A^\lambda (a|x ) P^\lambda_{B}(b_1 b_2 |y_1 y_2)
\end{equation}
with
\begin{equation}
 \sum_{b_2} P^\lambda_{B}(b_1 b_2| y_1 y_2 ) \qquad \text{independent of } y_2.
\end{equation}
Further, let $\mathsf{PostLoc}$ denote the set of sequential correlations $P$ that have a local hidden-variable model with respect to $A|B$ and whose post-selected correlations are local as well, i.e. for $P \in \mathsf{PostLoc}$ we have
\begin{equation}
\label{eq:standardlocal}
 P(a b_1 b_2 |x y_1 y_2)= \sum_\lambda p_\lambda P_A^\lambda (a|x ) P_{B}^\lambda(b_1b_2 |y_1 y_2)
\end{equation}
and
\begin{equation}
 P(a b_2| x y_2 b_1 y_1 )= \sum_\lambda p_{\lambda|b_1,y_1} P_A^\lambda(a|x )P_B^\lambda(b_2|y_2 ).
\end{equation}

Both $\mathsf{TOLoc}$ and $\mathsf{PostLoc}$ are convex polytopes, that is compact convex sets with a finite number of extreme points. By proposition \ref{prop:nohidden} correlations from $\mathsf{TOLoc}$ do not display hidden nonlocality, so that we have the inclusion $\mathsf{TOLoc} \subseteq \mathsf{PostLoc}$. Next we will show that this inclusion is in fact strict, i.e. there are correlations $P$ that are in $\mathsf{PostLoc}$ but not in $\mathsf{TOLoc}$.

In general a convex polytope can be either described by its extreme points or equivalently by the set of facet-defining half-spaces. These half-spaces are given by linear inequalities
\begin{equation}
\label{eq:inequality}
\beta(P)= \sum_{a,b_1,b_2,x,y_1,y_2} \beta_{ab_1b_2| x y_1 y_2} P(a b_1b_2 |x y_1y_2) \leq \tilde{\beta}.
\end{equation}
Using the polytope software {\textsc{porta}} \cite{Christof} we fully characterized the polytope $\mathsf{TOLoc}$ in terms of its facet-defining inequalities, see appendix \ref{appendix} for details. The problem of deciding whether $\mathsf{TOLoc} \subsetneq \mathsf{PostLoc}$ or $\mathsf{TOLoc} = \mathsf{PostLoc}$ can then be cast into a set of linear programs maximising the inequalities of $\mathsf{TOLoc}$ over probability distributions from $\mathsf{PostLoc}$. \rodrigo{Let us introduce explicitly one of the facet-defining inequalities that is of special relevance. In order to do it, let us employ a common parametrization of the joint probability distribution in terms of correlators. That is

\begin{equation}
\label{eq:probInCorr}
 \begin{aligned}
    P&(ab_1b_2|xy_1y_2)=\\
    &\frac{1}{8}\left[ 1 + a \ev{A_x} + b_1\ev{B^1_{y_1}} + b_2\ev{B^2_{y_1y_2}}\right.\\
          & + ab_1\ev{A_xB^1_{y_1}}+ab_2\ev{A_xB^2_{y_1y_2}}\\ &+b_1b_2 \ev{B^1_{y_1}B^2_{y_1y_2}}
+\left.ab_1b_2\ev{A_xB^1_{y_1}B^2_{y_1y_2}} \right],
 \end{aligned}
\end{equation}
where $\ev{A_x} = P(a=+1|x)-P(a=-1|x)$ is the expectation value of the outcome for party $A$ given input $x$, $\ev{A_x B^2_{y_1y_2}}= P(a\cdot b_2=+1|xy_1y_2)-P(a\cdot b_2=-1|xy_1y_2)$ is the expectation value of the product of the outcome of $A$ and the second outcome of $B$ given the inputs $x,y_1,y_2$, and so on.
% \begin{align}
%  \label{eq:correlators}
% A_x &= \sum_a a P_A (a|x)\\
% B_{y_1}^{(1)} &= \sum_{b_1}  b_1 P_{B_1}(b_1|y_1)\\
% B_{y_1,y_2}^{(2)} &= \sum_{b_2}  b_2  P_{B_2}(b_2|y_1 y_2)\\
% B_{y_1,y_2}^{(1,2)} &= \sum_{b_1 b_2}  b_1 b_2 P_{B_1B_2}(b_1 b_2|y_1 y_2)
% \end{align}

If one defines the following linear combinations of correlators
\begin{align}
\label{eq:newB0}
 B &=\frac{1}{2}\left[(1+B^1_0)B^2_{01} - (1-B^1_0)B^2_{00}\right]\\
\label{eq:newB1}
B' &= \frac{1}{2} \left[(1-B^1_1)B^2_{11} + (1+B^1_1)B^2_{10}\right],
\end{align}
one can write one of the facet inequalities of $\mathsf{TOLoc}$ as
\begin{equation}
 \label{eq:corrIneq}
\beta(P) :=\ev{A_0 (B -B') - A_1 (B+B')} \leq 2.
\end{equation}
With this parametrization, this inequality resembles the CHSH bipartite Bell inequality. As it turns out, the inequality \eqref{eq:corrIneq} can be violated by probability distributions from $\mathsf{PostLoc}$:
\begin{equation}
\label{eq:linprogseq}
\beta_\mathsf{PostLoc} =  \max_{\mathsf{PostLoc}} \beta(P) = 4,
\end{equation}
showing that $\mathsf{TOLoc}\subsetneq\mathsf{PostLoc}$. The correlations $P^* \in \mathsf{PostLoc}$ attaining the maximum in \eqref{eq:linprogseq} have by definition a standard local decomposition with respect to $A|B$ and do not display hidden nonlocality. However, the violation of \eqref{eq:corrIneq} by $P^*$ demonstrates that these correlations cannot be explained by a time-ordered local model for sequential correlations. One may wonder whether such correlations that belong to $\mathsf{PostLoc}$ but lie outside $\mathsf{TOLoc}$ are physical, in the sense that they can be obtained with measurements on quantum states. We show that this is the case by constructing an explicit example.

We consider a GHZ state $\frac{1}{\sqrt{2}}\left(\ket{000}+\ket{111}\right)$ shared between $A$ and $B$. The first qubit is sent to $A$ while $B$ has access to the other two qubits. At each round of the experiment $A$ measures one out
of two possible observables $A_x$ with $x \in \lbrace 0, 1\rbrace$. Part $B$ first measures one of his qubits according to the observables $B^1_{y_1}$ with $y_1 \in \lbrace 0, 1\rbrace$. Then, $B$ performs a measurement on its second 
qubit, using one out of eight possible observables $B^2_{b_1,y_1,y_2}$ that depend on $y_1$, the outcome $b_1$ and $y_2$. For such a configuration, we have numerically found a maximal violation of $\beta(P)=2\sqrt{2}$. 
To achieve this it is sufficient to consider all measurements to lie in the X-Z plane of the Bloch-sphere, that is, $O(\theta)=\cos{(\theta)}Z+\sin{(\theta)}X$ with $X$ and $Z$ being the usual Pauli matrices. Setting $\theta_{a_0}=\pi/2$,
$\theta_{a_1}=-\pi$, $\theta_{b^1_0}=\theta_{b^1_1}=\pi/2$, $\theta_{b^2_{0,0,1}}=-\theta_{b^2_{0,1,0}}=\theta_{b^2_{1,1,0}}=\theta_{b^2_{1,1,1}}=\pi/4$, $\theta_{b^2_{1,0,0}}=\theta_{b^2_{1,0,1}}=3\pi/4$ and 
$\theta_{b^2_{0,0,0}}=-\theta_{b^2_{0,1,1}}=\pi/3$ we achieve the optimal value of $\beta(P)=2\sqrt{2}$ while not violating the usual CHSH inequality (conditioned on any possible values of $y_1$ and $b_1$). 
It is surprising that the observables for the first measurement performed by $B$ can be the same, its only role is to prepare the state $\frac{1}{\sqrt{2}}\left(\ket{00}+(-1)^{b_1}\ket{11}\right)$
with which the rest of the experiment is to be performed.

}

Now, as correlations from $\mathsf{TOLoc}$ are known to be compatible with our operational definition of sequential locality, the violation \eqref{eq:linprogseq} raises the question whether there is a sequential wiring that takes $P^*$ to bipartite correlations $P'$ that are nonlocal in the standard sense. In fact, we can prove an even stronger result.

\begin{thm}\label{thm:easyscenario}
Consider the bipartite sequential correlation scenario with respect to $(1,2)$, where each measurement has binary inputs and outputs. Then $\mathsf{OpLoc}=\mathsf{TOLoc}$.
\end{thm}

\begin{proof}
That time-ordered local models are compatible with the operational definition, i.e. $\mathsf{TOLoc}\subset \mathsf{OpLoc}$,  was shown in proposition  \ref{thm:modelscomp}. To see the converse, consider $P$ to be compatible with the operational definition. We have that all post-selections have a local model
\begin{equation}
\label{eq:postcond}
P(ab_2|xy_2b_1y_1)= \sum_\lambda p_\lambda^{b_1|y_1} P_A^\lambda(a|x) P_B^\lambda(b_2|y_2)
\end{equation}
for $a,b_1,b_2 \in \lbrace -1,1\rbrace$ and $x,y_1,y_2 \in \lbrace 0,1 \rbrace$. Further, for all sequential wirings,  specified by functions $f_1,f_2,g$, the wired correlations
\begin{equation}
\label{eq:wiringcond}
P'(ac|xz) = \sum_{\substack{b_1,b_2\\ \text{s.t. }g(y_1,b_1,b_2)=c}} P(ab_1b_2|x f_1(z) f_2(z,b_1))
\end{equation}
are local as well. The conditions \eqref{eq:postcond} and \eqref{eq:wiringcond} are linear constraints on the probabilities of $P$, so that we can define  linear programs
\begin{equation}\label{eq:linprogs}
\begin{aligned}
\beta^\star=  &\text{ maximise } &&\beta(P) \\
&\text{ subject to} && P\text{ fulfils \eqref{eq:postcond} and \eqref{eq:wiringcond}},
\end{aligned}
\end{equation}
for all facet defining inequalities $\beta$ of $\mathsf{TOLoc}$. In the present case of just one measurement for $A$ and two for $B$, these conditions are still tractable and we solved the linear programs using the software {\textsc{yalmip}} \cite{Loefberg2004}. We find
\begin{equation}
\beta^\star = \max_{\mathsf{TOLoc}}\beta(P)
\end{equation}
for all facet defining inequalities $\beta$ of $\mathsf{TOLoc}$, which shows that the set of correlations compatible with the operational definition of sequential nonlocality is equal to $\mathsf{TOLoc}$.
\end{proof}

So, for this simple scenario, where $A$ performs a single measurement and $B$ a sequence of two with binary inputs and outputs for all of them, the time-ordered local models exactly capture the operational definition of locality. Correlations admitting a time-ordered local model are not only compatible with the allowed sequential operations, but having such a model is equivalent to be sequentially local in the operational sense.

This result together with the fact $\mathsf{TOLoc}\subsetneq\mathsf{PostLoc}$, shown above, implies that apart from Popescu's hidden nonlocality there is a new form of nonlocality that can be revealed by studying correlations arising in scenarios of measurement sequences.
Formally stated we have the
\begin{thm}
\label{thm:postbiggerop}
In the bipartite sequential scenario with respect to $(1,2)$ with binary inputs and outputs there exist correlations $P \in \mathsf{PostLoc}$ that are nonlocal in the operational sense, i.e. $\mathsf{OpLoc} \subsetneq \mathsf{PostLoc}$.
\end{thm}

\rodrigo{To clarify the interpretation of this result, let us consider an hypothetical scenario where a referee is given a device producing sequential correlations $P(a_1b_1b_2|x_1y_1y_2)$. 
The goal of the referee is to infer whether two parties $A$ and $B$ could use this device to violate a Bell inequality, which would make the device potentially useful for quantum key distribution, 
randomness generation or any other information protocol based on nonlocality. The first naive attempt of the referee would be to check if the device, without any processing, is capable of violating a 
Bell inequality between $A$ and $B$. After concluding that this is not the case and, aware of the notion of hidden-nonlocality by Popescu, he implements a protocol of postselection. That is, he discards the events 
in which the input and output of the first measurement of $B$ are different from a certain combination $b_1y_1$. After doing so, he checks whether this postselected statistics violate a Bell inequality. He finds a 
negative answer for every combination of $b_1y_1$ and concludes that there is nothing useful in this device as it is producing correlations that are classical in every sense. The implication of Thm. \ref{thm:postbiggerop} is 
simply that the referee may be mistaken in concluding so. There exist correlations that may seem useless for such referee's criteria, but that can be turned into nonlocal correlations by simply performing a local processing of 
information by $A$ and $B$.

The question that naturally arises is: what should have the referee checked to avoid any wrong conclusion? In this simple scenario of one dichotomic measurement for $A$ and two sequential 
dichotomic measurements for Bob, due to Thm. \ref{thm:easyscenario}, we know that he should have checked whether $P$ had a time-ordered local model. 
If the correlations have such a model, then they are useless for such purposes. If the correlations do not have such a decomposition, then he can be sure that some protocol allows for a 
Bell inequality violation between $A$ and $B$. Whether this last implication holds in a general scenario is an open question that we enunciate, among others, in the next section.}
%This is certainly a surprising result. There are correlations $P$ that, when seen as bipartite nonsignalling correlations, have a standard local model as in \eqref{eq:standardlocal} with respect to $A|B$ and do not display hidden nonlocality, however, they allow for the creation of nonlocal bipartite correlations by a local wiring on $B$'s side.

%Furthermore, the above result shows that all the different manifestations of nonlocality, be it the standard one, hidden nonlocality, or this new form of sequential nonlocality, can be detected by using the facet inequalities of $\mathsf{TOLoc}$. A violation of any of these inequalities certifies the presence of one or several forms of nonlocality. This simplifies the analysis considerably. Instead of having to check, for a given sequential correlation $P$, all post-selections and all protocols involving post-selection and sequential wirings, one only needs to check whether $P$ satisfies all facet inequalities of  $\mathsf{TOLoc}$ to decide if $P$ is local.

\section{Open questions}

The previous section, in particular theorem \ref{thm:postbiggerop}, shows that the study of nonlocality in sequential correlation scenarios does not reduce to the study of standard Bell nonlocality; on the contrary, inequivalent forms of nonlocality must be distinguished in these scenarios. However, a full characterisation of all forms of nonlocality is still elusive. In the following we will formulate and discuss several interesting questions that remain open.

One of the most interesting open questions with respect to sequential nonlocality concerns the relationship between the set $\mathsf{TOLoc}$ and the set of correlations that are sequentially local in the operational sense. We know that having a time-ordered local model implies being local in the operational sense, the converse, however, remains an open problem in the case of more general scenarios.

\begin{prob}
\label{prob:TOLoc}
In a general sequential correlation scenario, $\mathsf{TOLoc} \subsetneq \mathsf{OpLoc}$ or $\mathsf{TOLoc} = \mathsf{OpLoc}$?
\end{prob}

Suppose $\mathsf{TOLoc} = \mathsf{OpLoc}$. Then, for any sequential correlation scenario, the complicated set of operationally local correlations can be characterised by the facet inequalities corresponding to the set $\mathsf{TOLoc}$ and all types of nonlocality for this scenario can be detected by these inequalities. If, however, $\mathsf{TOLoc} \subsetneq \mathsf{OpLoc}$, then, for some scenario, there are sequential correlations that remain local under all protocols involving wirings and post-selection while lacking a time-ordered local model.

Another relevant open problem is related to the nonlocality
displayed by quantum states. Does this new form of nonlocality
open the possibility to detect more quantum states as nonlocal
than would be possible with standard Bell tests or using Popescu's
argument of hidden nonlocality? Due to the result of
\cite{Popescu1995} we know that there are quantum states with a
local hidden-variable model for all projective measurements that
display hidden nonlocality; furthermore \cite{Hirsch2013} provides
a class of entangled states that show hidden nonlocality while
having a local model for general measurements (POVMs). But are
there quantum states that do not display hidden nonlocality in a
given sequential scenario but nevertheless give rise to
correlations that do not have a time-ordered local model? This question
was also raised in \cite{Teufel1997}. If so, this would correspond
to a new form of nonlocality exhibited by quantum states going
beyond both standard and hidden nonlocality.

\begin{prob}
\label{prob:quantum}
Is there a quantum state $\varrho$ acting on a product Hilbert space $\mathfrak{H}_1 \otimes \mathfrak{H}_2$ such that the following holds?
\begin{itemize}
\item[(i)] All correlations obtained from single projective measurements on $\varrho$ are local.
\item[(ii)] All sequential correlations $P$ obtained by measurements on $\varrho$ do not display hidden nonlocality.
\item[(iii)] For one choice of quantum measurements the sequential correlations $P$ do not admit a time-ordered local model.
\end{itemize}
\end{prob}

Note that this problem is connected to the open question whether
generalised measurements in form of POVMs offer an advantage over
projective measurements for detecting standard nonlocality of
quantum states and also to problem \ref{prob:TOLoc}.

\begin{prop}
Assume $\mathsf{TOLoc} = \mathsf{OpLoc}$ and further that every quantum state $\varrho$ that has a standard local model for projective measurements also has such a model for measurements given by POVMs. Then the answer to problem \ref{prob:quantum} is negative.
\end{prop}

\begin{proof}
We want to show that under the given assumptions the conditions (i),(ii), and (iii) of problem \ref{prob:quantum} cannot be all satisfied. We assume (ii) and (iii) and show a contradiction with (i). Assuming $\mathsf{TOLoc} = \mathsf{OpLoc}$ implies that (iii) is equivalent to the existence of sequential correlations obtained from $\varrho$ that are sequentially nonlocal in the operational sense. Assuming (ii) excludes the possibility of hidden nonlocality, therefore only leaving the possibility that there are sequential wirings taking the correlations $P$ to some bipartite nonlocal correlations $P'$. Applying these wirings on the sequential correlations obtained from projective measurements on $\varrho$ defines a quantum realisation of the nonlocal correlations $P'$ with POVMs for both parties. Now, the assumption that POVMs do not offer any advantage over projections implies a contradiction with (i).
\end{proof}

Popescu already mentioned that his argument using projective measurements to reveal hidden nonlocality does not apply to the case of local dimension $d=2$ \cite{Popescu1995}. The states found in \cite{Gisin1996} in dimension $d=2$ do display hidden nonlocality when sequences of generalised measurements in form of POVMs are applied, however, these states do not have a standard local model for all measurements, but are only local in the sense that they do not violate the CHSH inequality for rounds of single measurements.

The authors of \cite{Teufel1997} further presented states in dimension $d\geq 3$ that fulfil conditions (i) and (ii) of \ref{prob:quantum}, but they were not able to conclude whether (iii) holds. Based on these findings and the conjecture that entanglement of a quantum state is equivalent to not having a time-ordered local model they also proposed a scheme for the classification of nonlocality. According to this scheme the nonlocality of a quantum state is characterised by two natural numbers $\langle N,n \rangle$, the \emph{indices of nonlocality.} The first index $N$ denotes the length of the sequences of measurements necessary to reveal the nonlocality, i.e. the smallest number such that the quantum state gives rise to correlations that do not have a time-ordered local model with respect to $(N,N)$. For instance, pure entangled states have $N=1$ and Werner states in dimension $d\geq 5$ have $N=2$; separable states have $N=\infty$. If $N< \infty$, then the second index $n$ denotes the smallest value of $N$ that can be attained by non-trivial measurements. For the case of Werner states in dimension $d\geq 5$ we have $n=1$, as the post-measurement state has $N=1$.
For states with $N = \infty$ the second index is defined as the minimal number of copies of the state needed to reveal its nonlocality.

%Concerning the nonlocality within quantum theory the ultimate goal would be to identify correlation scenarios in which any entangled state displays some sort of non-classical behaviour. With sequential correlation scenarios we presented one possibility to study nonlocality in situations more general than the standard Bell scenario and we further showed that a new form of nonlocality can arise.
%A different generalisation of the standard scenario was investigated in \cite{Fritz2012}.
%Building on previous work \cite{Branciard2012} he considers scenarios with several sources producing physical systems, but no measurement choice for the observers. Whereas every standard Bell scenario can be mapped into such a generalised correlation scenario, most of these correlation scenarios do not correspond to a standard Bell scenario; he describes examples of nonlocality in several of these generalised scenarios, while raising many open questions at the same time.

So far, all the problems that we have tackled in this paper are concerned with locality, classicality and the different definitions and relations that emerge in a scenario with sequential measurements. But similar questions can be posed by substituting locality or classicality by quantumness. Let us define analogous versions of $\mathsf{TOLoc}$ and $\mathsf{OpLoc}$ for the quantum case.

%%quantum versions...
\begin{dfn}[Time-ordered quantum]
Sequential correlations $P$ belong to the set $\mathsf{TOQuant}$ of time-ordered quantum correlations, if there are
\begin{itemize}
\item[(i)] a quantum state $\varrho$ on some product Hilbert space  $\mathfrak{H}_1 \otimes \mathfrak{H}_2$
\item[(ii)]  measurements on $\mathfrak{H}_1$ and $\mathfrak{H}_2$ given by the Kraus operators  $\lbrace K_{a_1}^{x_1}\rbrace_{a_1}$ and $\lbrace L_{b_1}^{y_1}\rbrace_{b_1}$ respectively
\item[(iii)] and projective measurements  $\lbrace M_{a_2}^{x_2}\rbrace_{a_2}$ and $\lbrace N_{b_2}^{y_2}\rbrace_{b_2}$ on $\mathfrak{H}_1$ and $\mathfrak{H}_2$
\end{itemize}
such that the correlations $P$ can be expressed as
\begin{multline}
\label{eq:seqqu}
P(a_1a_2b_1 b_2 |x_1 x_2 y_1 y_2)\\= \trace\left((K_{a_1}^{x_1})^\dagger M_{a_2}^{x_2}K_{a_1}^{x_1}\otimes (L_{b_1}^{y_1})^\dagger N_{b_2}^{y_2}L_{b_1}^{y_1} \varrho\right).
\end{multline}
\end{dfn}

Just as before one can consider the correlations resulting from some protocol and ask whether there is a quantum realisation for this final bipartite distribution, i.e. whether
\begin{equation}
\label{eq:quantumreal}
P'(ab|xy)= \trace\left(M_{a}^{x}\otimes N_{b}^{y} \varrho\right)
\end{equation}
for some quantum state $\varrho$ and measurements $M_a^x$ and $N_b^y$. Thus, analogous to definition \ref{dfn:sequentialLocality} one can define the set of operationally quantum correlations.

\begin{dfn}
The set of \emph{operationally quantum correlations} $\mathsf{OpQuant}$ is the set of sequential correlations $P$ such that for any $n$ the product correlations $P^{\times n}$ are mapped by any valid protocol to bipartite correlations that admit a quantum realisation.
\end{dfn}

By a reasoning similar to the one used in proposition \ref{thm:modelscomp} it is clear that we have the inclusion $\mathsf{TOQuant} \subset \mathsf{OpQuant}$. As a quantum version of open problem \ref{prob:TOLoc} we can then pose the following

\begin{prob}
Are there correlations that are operationally quantum but do not belong to the set of time-ordered quantum correlations? That is, more formally, do we have
$\mathsf{TOQuant} \subsetneq \mathsf{OpQuant}$ or $\mathsf{TOQuant} = \mathsf{OpQuant}$?
\end{prob}

Let us discuss the two possibilities separately. Consider that $\mathsf{TOQuant}\subsetneq\mathsf{OpQuant}$. In this case, there are probability distributions that (i) do not have a decomposition of the form \eqref{eq:seqqu} and (ii) result after any valid protocol in correlations with a quantum realisation of the form \eqref{eq:quantumreal}.
Clearly, (i) implies that one cannot obtain these  correlations by performing measurements on quantum states. However, all the correlations (in the sense of a probability distribution between distant observers, without any temporal order between measurements) that can be generated out of them, are quantum. This implies that if one attempts to characterise the observable statistics valid within quantum theory, it will not suffice to characterise the correlations between distant observers, but also scenarios with sequences of measurements need to be considered. This would suggest that attempts to define quantum correlations via information principles might leave a non-trivial part of quantum theory aside if they do not consider sequential measurements.
On the other hand, if $\mathsf{TOQuant}=\mathsf{OpQuant}$ a converse reasoning applies. It would be striking that both sets are equivalent, since the constraints to define $\mathsf{TOQuant}$ appear to be stronger that the ones of $\mathsf{OpQuant}$. Note that the decomposition  \eqref{eq:seqqu} demands a well-defined post-measurement state, whereas the decomposition \eqref{eq:quantumreal} only requires the validity of the Born rule. Therefore, the equality of the two sets would support the idea that the whole set of observable statistics according to quantum theory only depends on the state space and measurements together with the Born rule, rather than the transformations of states after measurements.

\section{Conclusions}
To summarise, we have studied nonlocality in scenarios where the parties are allowed to perform sequences of measurements. As we have seen, sequential correlations give rise to inequivalent notions of nonlocality that we summarise in the following.

\begin{enumerate}
\item $\mathsf{PostLoc}$: This is the set of probability distributions that are local in the standard bipartite sense studied by Bell and that do not show any  hidden nonlocality, i.e. the probability distributions after post-selection on the previous measurements are local.
\item $\mathsf{OpLoc}$: This is the  set of sequential correlations such that any valid protocol takes an arbitrary number of copies of these correlations to standard bipartite local correlations. These protocols process classical information locally and correspond to the allowed operations in a locality scenario, i.e. nonlocality is the resource that cannot be created using these operations.
\item $\mathsf{TOLoc}$ This is the set of probability distributions that have a time-ordered local model, i.e. a local hidden-variable model that respects the causal structure of the sequential correlation scenario.
\end{enumerate}

We have studied the different relations between these sets and our findings can be summarised as

\begin{equation}\label{eq:conclusion}
\mathsf{TOLoc}\subseteq \mathsf{OpLoc}
\subsetneq\mathsf{PostLoc},
\end{equation}
where we could show the equality $\mathsf{TOLoc}= \mathsf{OpLoc}$ in a simple scenario. We further stated and discussed several interesting open problems, among which the most important one concerns the  question whether $\mathsf{TOLoc}= \mathsf{OpLoc} $ or $\mathsf{TOLoc}\subsetneq \mathsf{OpLoc}$ in the general case. Furthermore, we presented a quantum version of the different sets leading to analogous questions concerning the quantumness of correlations.

\section{Acknowledgements}

We thank T. Vert\'esi and Ll. Masanes for insightful discussions. This work is supported by the ERC grant TAQ, the ERC Starting Grant PERCENT, the
Spanish projects FIS2010-14830 and CHIST-ERA DIQIP, the
Generalitat de Catalunya, and the Excellence Initiative of the German Federal and State Governments (Grant ZUK 43). MN acknowledges support from the John Templeton Foundation.

%\bibliography{seq}
\appendix
\section{Inequalities of $\mathsf{TOLoc}$ for a simple scenario}
\label{appendix}

Using standard techniques of polytope analysis one can easily obtain all the facets of the set $\mathsf{TOLoc}$. This is however a computationally costly task. We have been able to obtain all the facets only for the scenario of two parties, $A$ and $B$, where $A$ performs chooses one binary measurement out of two at each round; and $B$ performs two sequential binary measurements choosing also at each time between two different measurements. The experiment is described by the probability distribution $P(a_1b_1b_2|x_1y_1y_2)$.

The facets can be divided into three groups. First, inequalities involving only the marginal $P(a_1b_1|x_1y_1)$ or $P(a_1b_2|x_1y_1y_2)$. The ones of the former type are equivalent (up to symmetries of permutation of inputs, outputs and parties) to the well-known CHSH inequality \cite{Clauser1969},

\begin{multline}
\label{eq:chsh}
P(a_1=b_1|00)+P(a_1=b_1|01)\\+P(a_1=b_1|10)+P(a_1\neq b_1|11)\leq 3.
\end{multline}
The inequalities involving $P(a_1b_2|x_1y_1y_2)$ are also equivalent to the CHSH. Now $B$ can choose among four different inputs given by $(y_1,y_2)$, these inequalities correspond to a \emph{lifting} \cite{Pironio2005} of the CHSH inequalities and are given by

\begin{multline}
\label{eq:liftedchsh}
P(a_1=b_2|000)+P(a_1=b_2|011)\\+P(a_1=b_2|100)+P(a_1\neq b_1|111)\leq 3.
\end{multline}
and its symmetries.

Secondly, there are facets involving also conditional probability distributions of the kind $P(a_1b_2|x_1y_2b_1y_1)$. These facets are again equivalent to the CHSH inequality, but now conditioned on a certain input and output $(y_1,b_1)$. For example
\begin{multline}
\label{eq:conditionedchsh}
P(a_1=b_1|00,b_1y_1)+P(a_1=b_1|01,b_1y_1)\\+P(a_1=b_1|10,b_1y_1)+P(a_1\neq b_1|11,b_1y_1)\leq 3
\end{multline}
for all possible combinations of $(b_1,y_1)$ and also all symmetries.

Lastly, there is a third kind of facets that involve the whole probability distribution $P(a_1b_1b_2|x_1y_1y_2)$. Note that the inequalities of the first kind are essentially conditions ensuring standard bipartite locality \eqref{eq:HiddenNonlocality1}, those of the second kind ensuring that there is no hidden nonlocality. The inequalities of the third group, however, are related to a different notion of nonlocality that arises in the sequential scenario. A representative of this third class of inequalities is given by \eqref{eq:corrIneq}.

\end{document}

%% file: sequence.pdf_t
\begin{picture}(0,0)%
\includegraphics{sequence.pdf}%
\end{picture}%
\setlength{\unitlength}{4144sp}%
\begingroup\makeatletter\ifx\SetFigFont\undefined%
\gdef\SetFigFont#1#2#3#4#5{%
  \reset@font\fontsize{#1}{#2pt}%
  \fontfamily{#3}\fontseries{#4}\fontshape{#5}%
  \selectfont}%
\fi\endgroup%
\begin{picture}(3284,1647)(5289,-976)
\put(5491,-916){\makebox(0,0)[lb]{\smash{{\SetFigFont{10}{12.0}{\familydefault}{\mddefault}{\updefault}{\color[rgb]{0,0,0}$a_2$}%
}}}}
\put(8281,524){\makebox(0,0)[lb]{\smash{{\SetFigFont{10}{12.0}{\familydefault}{\mddefault}{\updefault}{\color[rgb]{0,0,0}$y_2$}%
}}}}
\put(8281,-916){\makebox(0,0)[lb]{\smash{{\SetFigFont{10}{12.0}{\familydefault}{\mddefault}{\updefault}{\color[rgb]{0,0,0}$b_2$}%
}}}}
\put(6256,524){\makebox(0,0)[lb]{\smash{{\SetFigFont{10}{12.0}{\familydefault}{\mddefault}{\updefault}{\color[rgb]{0,0,0}$x_1$}%
}}}}
\put(6256,-916){\makebox(0,0)[lb]{\smash{{\SetFigFont{10}{12.0}{\familydefault}{\mddefault}{\updefault}{\color[rgb]{0,0,0}$a_1$}%
}}}}
\put(7516,-916){\makebox(0,0)[lb]{\smash{{\SetFigFont{10}{12.0}{\familydefault}{\mddefault}{\updefault}{\color[rgb]{0,0,0}$b_1$}%
}}}}
\put(7516,524){\makebox(0,0)[lb]{\smash{{\SetFigFont{10}{12.0}{\familydefault}{\mddefault}{\updefault}{\color[rgb]{0,0,0}$y_1$}%
}}}}
\put(5491,524){\makebox(0,0)[lb]{\smash{{\SetFigFont{10}{12.0}{\familydefault}{\mddefault}{\updefault}{\color[rgb]{0,0,0}$x_2$}%
}}}}
\end{picture}%

%% file: causal.pdf_t
\begin{picture}(0,0)%
\includegraphics{causal.pdf}%
\end{picture}%
\setlength{\unitlength}{4144sp}%
\begingroup\makeatletter\ifx\SetFigFont\undefined%
\gdef\SetFigFont#1#2#3#4#5{%
  \reset@font\fontsize{#1}{#2pt}%
  \fontfamily{#3}\fontseries{#4}\fontshape{#5}%
  \selectfont}%
\fi\endgroup%
\begin{picture}(1754,1392)(4614,-2242)
\put(5446,-2131){\makebox(0,0)[lb]{\smash{{\SetFigFont{10}{12.0}{\familydefault}{\mddefault}{\updefault}{\color[rgb]{0,0,0}$\lambda$}%
}}}}
\put(4726,-1816){\makebox(0,0)[lb]{\smash{{\SetFigFont{10}{12.0}{\familydefault}{\mddefault}{\updefault}{\color[rgb]{0,0,0}$a$}%
}}}}
\put(4726,-1051){\makebox(0,0)[lb]{\smash{{\SetFigFont{10}{12.0}{\familydefault}{\mddefault}{\updefault}{\color[rgb]{0,0,0}$x$}%
}}}}
\put(6166,-1051){\makebox(0,0)[lb]{\smash{{\SetFigFont{10}{12.0}{\familydefault}{\mddefault}{\updefault}{\color[rgb]{0,0,0}$y$}%
}}}}
\put(6166,-1816){\makebox(0,0)[lb]{\smash{{\SetFigFont{10}{12.0}{\familydefault}{\mddefault}{\updefault}{\color[rgb]{0,0,0}$b$}%
}}}}
\end{picture}%

%% file: causalseq.pdf_t
\begin{picture}(0,0)%
\includegraphics{causalseq.pdf}%
\end{picture}%
\setlength{\unitlength}{4144sp}%
\begingroup\makeatletter\ifx\SetFigFont\undefined%
\gdef\SetFigFont#1#2#3#4#5{%
  \reset@font\fontsize{#1}{#2pt}%
  \fontfamily{#3}\fontseries{#4}\fontshape{#5}%
  \selectfont}%
\fi\endgroup%
\begin{picture}(3014,1626)(4344,-2476)
\put(6480,-1051){\makebox(0,0)[lb]{\smash{{\SetFigFont{10}{12.0}{\familydefault}{\mddefault}{\updefault}{\color[rgb]{0,0,0}$y_1$}%
}}}}
\put(6480,-1819){\makebox(0,0)[lb]{\smash{{\SetFigFont{10}{12.0}{\familydefault}{\mddefault}{\updefault}{\color[rgb]{0,0,0}$b_1$}%
}}}}
\put(7107,-1051){\makebox(0,0)[lb]{\smash{{\SetFigFont{10}{12.0}{\familydefault}{\mddefault}{\updefault}{\color[rgb]{0,0,0}$y_2$}%
}}}}
\put(7107,-1816){\makebox(0,0)[lb]{\smash{{\SetFigFont{10}{12.0}{\familydefault}{\mddefault}{\updefault}{\color[rgb]{0,0,0}$b_2$}%
}}}}
\put(5800,-2359){\makebox(0,0)[lb]{\smash{{\SetFigFont{10}{12.0}{\familydefault}{\mddefault}{\updefault}{\color[rgb]{0,0,0}$\lambda$}%
}}}}
\put(4410,-1048){\makebox(0,0)[lb]{\smash{{\SetFigFont{10}{12.0}{\familydefault}{\mddefault}{\updefault}{\color[rgb]{0,0,0}$x_2$}%
}}}}
\put(5040,-1051){\makebox(0,0)[lb]{\smash{{\SetFigFont{10}{12.0}{\familydefault}{\mddefault}{\updefault}{\color[rgb]{0,0,0}$x_1$}%
}}}}
\put(5040,-1816){\makebox(0,0)[lb]{\smash{{\SetFigFont{10}{12.0}{\familydefault}{\mddefault}{\updefault}{\color[rgb]{0,0,0}$a_1$}%
}}}}
\put(4410,-1819){\makebox(0,0)[lb]{\smash{{\SetFigFont{10}{12.0}{\familydefault}{\mddefault}{\updefault}{\color[rgb]{0,0,0}$a_2$}%
}}}}
\end{picture}%